\documentclass[11pt]{article}
\usepackage{amsmath}%
\usepackage{amsthm}%
\usepackage{mathptmx}%

\title{Recursion-Theoretic Ranking and Compression}

\date{October 4, 2016;
revised December 3, 2017}
\author{Lane A. Hemaspaandra and Daniel Rubery\thanks{Current address: 
Google Inc., Dan Rubery (drubery), 1600 Amphitheatre Parkway, Mountain View, CA 94043.}\\
        Department of Computer Science \\
        University of Rochester \\
        Rochester, NY 14627, USA}
\showhyphens{}

\newcommand{\pair}[1]{\mathopen\langle{#1}\mathclose\rangle}

\newcommand{\calf}{\ensuremath{{\cal F}}}
\newcommand{\calc}{\ensuremath{{\cal C}}}
\newcommand{\sigmastar}{\ensuremath{\Sigma^\ast}}
\newcommand{\sharpp}{{\rm \#P}}
\newcommand{\condition}{\mid}
\newcommand{\deltatwozero}{\ensuremath{\Delta_2^0}}
\newcommand{\domain}{\ensuremath{\mathrm{domain}}}

\newcommand{\nextstring}{\ensuremath{\mathrm{successor}}}

\newcommand{\frec}{\ensuremath{\mathrm{F}_\mathrm{REC}}}
\newcommand{\fre}{\ensuremath{\mathrm{F_\mathrm{PR}}}}
\newcommand{\rec}{\ensuremath{\mathrm{REC}}}
\newcommand{\K}{\ensuremath{\mathrm{K}}}
\newcommand{\re}{\ensuremath{\mathrm{RE}}}
\newcommand{\core}{\ensuremath{\mathrm{coRE}}}
\newcommand{\leqlex}{\mathbin{\leq_{\mathrm{lex}}}}

\newcommand{\glex}{\mathbin{>_{\mathrm{lex}}}}
\newcommand{\llex}{\mathbin{<_{\mathrm{lex}}}}
\newcommand{\oplushat}{\mathbin{\widehat{\oplus}}}
\def\tcomp{\text{-compressible}}
\def\trank{\text{-rankable}}
\newtheorem{theorem}{Theorem}[section]
\newtheorem{definition}[theorem]{Definition}
\newtheorem{corollary}[theorem]{Corollary}
\newtheorem{proposition}[theorem]{Proposition}

\begin{document}
\sloppy

\maketitle

\begin{abstract}
  For which sets $A$ does there exist a mapping, computed by a 
  total or partial recursive function, such that the mapping, when its domain
  is restricted to $A$, is a 1-to-1, onto mapping to $\sigmastar$?  And
  for which sets $A$ does there exist such a mapping that respects
  the lexicographical ordering within $A$?  Both cases are types of perfect,
  minimal hash functions.  The complexity-theoretic versions of these
  notions are known as compression functions 
  and ranking functions.  The present paper defines and
  studies the recursion-theoretic versions of compression and ranking
  functions, and in particular studies the question of which sets have, or
  lack, such functions.  
  Thus, this is a case where, in contrast to the usual direction of 
  notion 
transferal, 
  notions from 
   complexity theory are inspiring notions, and an investigation,
  in computability theory.

We show that the rankable and compressible sets 
  broadly populate the 1-truth-table degrees, and we prove that 
  every nonempty coRE cylinder is recursively compressible.
  \end{abstract}

\section{Introduction}\label{s:introduction}
This paper studies the recursion-theoretic case of how hard it is to
squeeze the air (more concretely, the elements of its complement) out
of a set $A$.  That is, we want to, by a total recursive function or a
partial recursive function, map in a 1-to-1, onto fashion from $A$ to
$\sigmastar$.  So our function, when viewed as being restricted to the
domain $A$, is a bijection between $A$ and $\sigmastar$.  In effect, 
each string
in $A$ is given a unique ``name'' (string) from $\sigmastar$, and
every ``name'' from $\sigmastar$ is used for some string in $A$.  
As has been pointed
out for the complexity-theoretic analogue
(where we are interested not in total and partial recursive
functions, but in polynomial-time functions), such functions are 
the analogues for infinite sets of perfect (i.e., no collisions among 
elements in $A$), 
minimal (i.e., every element of $\sigmastar$ is hit by some element of $A$)
hash functions, 
and are called 
\emph{compression functions}~\cite{gol-hem-kun:j:address}.

A particularly dramatic type of such function would be one that maps
from the $i$th element of an infinite set $A$ to the $i$th element of
$\sigmastar$.  Such a function---a
\emph{ranking function}---has all the above properties and in
addition respects the (lexicographical) ordering of the elements of
$A$.  
For the case of
polynomial-time functions, 
this type of issue was first studied by 
Allender~\cite{all:thesis:invertible} 
and Goldberg and
Sipser~\cite{gol-sip:cOutByJour:compression,gol-sip:j:compression}
a quarter of a century ago.

That seminal work
of Allender, Goldberg, and 
Sipser led other researchers to bring a closer lens to
the issue of what behavior the ranking function would be required to
have on inputs that did not belong to $A$~\cite{hem-rud:j:ranking}, to
study more flexible notions such as the abovementioned compression
functions~\cite{gol-hem-kun:j:address} and what are known as
scalability~\cite{gol-hom:j:scalability} and
semi-ranking~\cite{hem-ogi-zak-zim:j:psr-weak-P-rank}, and to study
ranking of extremely simple sets
(\cite{huy:j:rank,alv-jen:j:logcount},
see also \cite{all:thesis:invertible,gol-sip:j:compression}).
Even the original paper of Goldberg and Sipser already established that 
there are P sets whose ranking function is complete 
for the counting version of NP (namely $\sharpp$), i.e., 
that quite simple sets can have quite complex ranking functions.

The present paper 
studies compression and ranking in
their recursion-theoretic analogues.  These basically are the same
problems as in the complexity-theoretic case, except instead of
studying what can (and cannot) be done by polynomial-time functions,
we study what can (and cannot) be done by total recursive functions
and partial recursive functions.
The direction of studying ranking and compression by total
recursive functions 
was 
previously 
mentioned as an open direction in the conclusions
section of~\cite{gol-hem-kun:j:address}, which 
observed without 
proof what here are 
Theorem~\ref{t:109.2}
and Corollary~\ref{c:4.1/4.2-dan}.

\emph{Why} do we study this?  After all, programmers are not clamoring
to have recursion-theoretic perfect, minimal hash functions for
infinite sets.  But our motivation is not about satisfying a
programming need.  It is about learning more about the structure of
sets, and the nature of---and in some cases the impossibility
of---compression done by total and partial recursive functions.  
In particular,
what classes of sets can we show to have, or not have, such
compression and ranking functions?

Among the results are the following.
\begin{itemize}
\item Every 1-truth-table degree except the zero degree contains both
  sets that are recursively rankable and sets that are not recursively
  rankable (Theorems~\ref{t:103.1} and~\ref{t:123.1}).
  (So some recursively rankable sets are undecidable, and some 
  even fall outside
  of the arithmetical hierarchy.)

\item Every 1-truth-table degree except the zero degree
  contains some set that is recursively compressible yet is not
  recursively rankable (Theorem~\ref{t:123.1}).

\item Every nonempty coRE cylinder is recursively compressible
(Theorem~\ref{t:dan-160522-Beta1}), and it follows that 
all coRE-complete sets 
(see Corollary~\ref{c:4.1/4.2-dan})
and all nonempty coRE index sets 
(Corollary~\ref{c:dan-160522-Beta2})
are recursively compressible
However, no RE-complete set or 
coRE-complete set is 
recursively rankable
or even partial-recursively
rankable (Corollary~\ref{c:25.3BK54} and Corollary~\ref{c:167.3-dan}).

\item There are infinite
$\deltatwozero$ sets that are not even partial-recursively 
compressible (Theorem~\ref{t:151.2(b)-dan}).

\item Although every recursively compressible RE set is recursive
  (see Theorem~\ref{t:109.2}), each infinite set in $\re - \rec$ is an
  example of a partial-recursively compressible RE set that is not
  recursive or recursively compressible 
  (Proposition~\ref{p:106.1} and
  Corollary~\ref{c:111.15}).
  So although all coRE-complete sets are
  recursively compressible, no RE-complete set is recursively
  compressible.

\end{itemize}

\section{Related Work}\label{s:related}
The most closely related papers are those mentioned in 
Section~\ref{s:introduction}.  
Given the importance to this paper of mappings that are 
onto $\sigmastar$, we 
mention also the 
line of work, dating back to Brassard, Fortune, and Hopcroft's
early paper on one-way functions~\cite{bra-for-hop:t:one-way}, 
that looks at
the complexity of inverting
functions that map onto
$\sigmastar$~\cite{bra-for-hop:t:one-way,fen-for-nai-rog:j:inverse,hem-rot-wec:j:hard-certificates,rot:thesis-habilitation:certificates}.
However, both that line and the papers mentioned 
in Section~\ref{s:introduction} are 
about complexity-theoretic functions,
while in contrast the current paper is about recursion-theoretic functions.

In fact, this paper is quite 
the reverse of the typical direction of
inspiration.  A 
large number of the core concepts of complexity theory are defined by
direct analogy with notions from recursive function theory.  As just a
few examples, NP, the polynomial 
hierarchy~\cite{mey-sto:c:reg-exp-needs-exp-space,sto:j:poly}, most of complexity
theory's reduction notions~\cite{lad-lyn-sel:j:com}, 
(complexity-theoretic) creative/simple/immune/bimmune
sets~\cite{ber:c:complete-sets,bal-sch:j:immune,hom:j:simple,tor:thesis:relativized-hierarchies,hom-maa:j:oracle-lattice}, 
and the semi-feasible sets~\cite{sel:j:ana}
are lifted quite directly from recursive function theory, with, as
needed, the appropriate, natural changes to focus on the deterministic
and nondeterministic polynomial-time realms.  The debt that complexity
theory owes to recursive function theory is huge.  

Far less
common is for notions defined in complexity to then be studied
recursion-theoretically.  However, this paper is a small example of that, 
since it is taking the line of ranking/compression work started by
Allender, 
Goldberg, and Sipser in the 1980s and asking the same type of questions
in the setting of total and partial recursive functions.

The 
notions of retraceable sets,
regressive sets,
and isolic reductions
are the closest existing concepts in recursive function 
theory to the 
notions of rankable and compressible sets.  We now 
discuss each, pointing out how the notions differ from ours.

A set $A \subseteq \sigmastar$ is called \emph{regressive} 
if there exists an enumeration (note that the definition does not require that it be a recursive enumeration) 
of $A$ without repetitions $\{a_0, a_1, a_2, ... \}$, and a partial recursive function $f$ such that:
$ f(a_{n+1}) = a_n$
and
$ f(a_0) = a_0$~\cite{dek:c:isols-regressive}.
The set $A$ is called \emph{retraceable} if it meets the 
definition of regressive 
with respect to a (not necessarily recursive) 
enumeration that follows the standard lexicographic 
order~\cite{dek-myh:j:retraceable}. 
Odifreddi~\cite{odi:b:classical-recursion-theory}
comments that there is a ``surface analogy'' that 
r.e.\ is to recursive as regressive is to retraceable.  We similarly
mention that there is a surface analogy that 
$\fre$-compressible is to $\fre$\trank~as 
regressive is to retraceable.
We claim (and it is not too hard to see; one basically checks whether 
the input is $a_0$---which will be hardcoded into the program---and 
if not tries repeatedly applying $f$ until, if 
ever, one reaches $a_0$, keeping track of how many applications 
that took)
that 
each retraceable set is $\fre$\trank\ and 
each 
infinite regressive set is $\fre$-compressible.
We further claim that 
each set that is retraceable under a recursive retracing 
function $f$ 
is $\frec$\trank~(the same approach sketched above works, 
along with observing that if at any point in the $f$
application chain starting at a string $x$
the recursive retracing function maps a string $y \neq a_0$ 
to a string 
lexicographically
greater than 
or equal to $y$, then 
our original string is definitely not in the set and our 
$\frec$-ranking function
can output any value it likes as $x$'s purported rank).
We claim that the converses of these statements fail
rather dramatically; our notions are far more general.
For example, 
there are $\frec$\trank\ sets that are not retraceable
and indeed that are not even regressive
(and recall that the definitions of retraceable and regressive 
are 
with respect to partial recursive retracing functions, so this
is a very strong type of separation).%
\footnote{Let $s_0, s_1, s_2, ...$ 
enumerate $\sigmastar$ in lexicographic order.
Let $K$ be the RE-complete set, 
$\{x \condition x \in L(M_x)\}$.
Define $A = \{ s_{3i} \condition i \geq 0 \} \cup \{s_{3i+1} \condition i \in \K \} \cup \{s_{3i+2} \condition i \in \overline{\K}\}$.
Then $A$ is easily seen to be $\frec$\trank.
Yet we claim that 
$A$ is not retraceable and indeed is not even regressive.
(And the definitions of retraceable and regressive are 
with respect to partial recursive retracing functions, so this is 
even stronger than the claims that $A$ is not retraceable or 
not regressive
via some recursive retracing function).
Why is $A$ not regressive?  
A set is said to be 
\emph{immune} (or \emph{r.e.-immune}) if it is infinite but 
contains no infinite r.e.\ (equivalently, recursive) subsets.
Every regressive set is either r.e.\ or 
immune~\cite[Prop.~II.6.8]{odi:b:classical-recursion-theory}.
Yet $A$ is not r.e.~($\overline{X}$ clearly recursive 
many-one reduces to $A$) and $A$ is not immune (due to the 
having the recursive subset 
$\{s_{3i} \condition i \geq 0 \}$).
Thus $A$ is not 
regressive.
}

The notion of rankability is, in fact, so nonrestrictive 
that, as this paper will establish, 
every 1-truth-table degree contains an $\frec$\trank\ set.
What about the retraceable sets?  They are known to populate the
truth-table degrees.\footnote{The literature reference for this is a
  bit tricky.  Odifreddi~\cite[Proposition
  II.6.13]{odi:b:classical-recursion-theory} proves the result of
  Dekker and Myhill~\cite{dek-myh:j:retraceable} that each Turing
  degree contains a retraceable set.  However, the given proof in fact
  establishes that each truth-table degree contains a retraceable set,
  and that fact clearly is known to Odifreddi since at the start of
  Exercise VI.6.16.b he quietly attributes to Proposition II.6.13 the
  fact that each truth-table degree other than the zero degree
  contains an immune, retraceable
  set~\cite[p.~600]{odi:b:classical-recursion-theory}.}  
However, we can prove that, unlike 
the $\frec$\trank\ sets,
they do not populate the 1-truth-table 
degrees.\footnote{We state that as the following theorem.
The \emph{1-truth-table upward (reducibility) cone}
of a set $L$ is 
$\{ L' \condition L \leq_{1\hbox{-}tt} L'\}$ (the term is more 
commonly used for degrees~\cite{odi:b:classical-recursion-theory} although
the difference is inconsequential).
\begin{theorem}
There is a 1-truth-table degree (indeed, there is a 1-truth-table
upward cone) that contains no retraceable set.
\end{theorem}
\begin{proof}
For clarity, we first 
discuss 
the 1-truth-table degree case.
Recall that a set is said to be 
\emph{immune} (or \emph{r.e.-immune}) if it is infinite but 
contains no infinite r.e.\ (equivalently, recursive) subsets.
Since 
every retraceable set is recursive or immune~\cite{dek-myh:j:retraceable},
it will suffice to find a 1-truth-table degree
with no recursive or immune sets.
Let $A$ and $B$ be a pair of disjoint, r.e.\ sets 
that are recursively inseparable.  It is well-known that 
such pairs exist, e.g., $\{i \condition M_i(i)$ halts and outputs 1$\}$
and $\{i \condition M_i(i)$ halts and outputs 2$\}$.
We claim that if $A \leq_{1\hbox{-}tt} S$, then $S$ is not immune.
Let $f$ be a 1-truth-table 
reduction from $A$ to $S$ and suppose, seeking a 
contradiction, that $S$ is immune.
Then the set $L_A = \{f(x) \condition x \in A $ and $ f(x) $ uses 
the identity truth-table$\}$ is an r.e.\ subset of $S$, so $L_A$ it must be finite. 
Similarly, $L_B = \{ f(x) \condition x \in B $ and $f(x)$ uses the 
negation truth-table$\}$ must be finite.
We can use $L_A$ and $L_B$, however, to recursively separate $A$ 
and $B$. Define the function $g(x)$ as follows:

If $f(x)$ uses the identity truth-table, and $f(x) \in L_A$, then $g(x) = 1$.

If $f(x)$ uses the identity truth-table, and $f(x) \not \in L_A$, 
then $g(x) = 0$.

If $f(x)$ uses the negation truth-table, and $f(x) \in L_B$, then $g(x) = 0$.

If $f(x)$ uses the negation truth-table, and $f(x) \not \in L_B$, then
$g(x) = 1$.

For every $x \in A$, $g(x) = 1$, and for every $x \in B$, $g(x) =
0$. The function $g$ is recursive because both $L_A$ and $L_B$ are
finite. However, this contradicts the recursive inseparability of $A$
and $B$, so $S$ must not be immune. Thus the 1-truth-table degree of $A$
contains no recursive or immune sets, and so contains no retraceable
sets.

The proof as given above in fact also shows that the 
1-truth-table upward (reducibility) cone of $A$ contains 
no retraceable sets.
\end{proof}}

Another concept from recursive function theory that has a 
similar flavor 
to the notions we are looking at is the 
notion of an \emph{isolic reduction}.
$A$ is said to \emph{isolic-reduce} to $B$ 
if there exists a one-to-one partial recursive function
$f$ such that $A = f^{-1}(B)$ (see~\cite[p.~124]{rog:b:rft}).
Sets that isolic-reduce
to $\sigmastar$ thus have a similar definitional flavor to 
our notion of $\fre$-compression.
However, 
note that isolic reductions are required to be one-to-one, 
and so unlike our notions cannot allow even a single element of 
$\overline{A}$ to be mapped to a member of $B$.
That is enough to make them strikingly 
differ in behavior from
the notions we
are studying.   In particular, 
our paper puts recursively- and partial
recursively-compressible sets into \emph{every} 1-truth-table degree, but 
in contrast
the class of sets that
isolic-reduce to $\sigmastar$ is precisely the infinite r.e.\ sets, and
so \emph{no} sufficiently hard 1-truth-table degrees contain any sets 
that 
isolic-reduce to $\sigmastar$.

\section{Definitions}\label{s:definitions}

Throughout this paper, we fix the alphabet to be the binary alphabet
$\Sigma = \{0,1\}$.  So all our notions will involve (total or
partial) functions whose input universe is $\sigmastar$ and whose
codomain is $\sigmastar$, and all classes (e.g., the recursive sets)
are viewed as being over sets whose alphabet is $\Sigma$.  

Why is it natural to focus just on $\Sigma=\{0,1\}$?  For every two
finite alphabets $\Sigma'$ and $\Sigma''$, there is a recursive,
order-respecting bijection between ${\Sigma'}^\ast$ and ${\Sigma''}^\ast$.
So for all natural purposes in the context of recursive function
theory, any pair of finite alphabets are essentially computationally
interchangeable.

$\frec$ will denote the class of all total recursive functions from
$\sigmastar$ to $\sigmastar$.  $\fre$ will denote the class of all
partial recursive functions from $\sigmastar$ to $\sigmastar$.  
$\domain(f)$ is the set of inputs on which a (potentially partial)
function $f$ is not undefined, e.g., if $f$ is a total function,
$\domain(f) = \sigmastar$.

$\rec$ and $\re$ will denote the recursive sets and the recursively
enumerable sets.  As usual,
$\core = \{A \condition \overline{A} \in \re\}$ and $\deltatwozero$
the class of all sets $A$ such that there exists a set $B \in \re$
such that $A$ is recursive in $B$ (i.e., $A$ recursively Turing
reduces to $B$).  These are low levels of what is known
as the 
\emph{arithmetical (or Kleene--Mostowski) hierarchy}~\cite{kle:j:arithmetical,mos:j:arithmetical}.
We will often use r.e.\ and co-r.e.\ as adjectival 
forms of $\re$ and $\core$, e.g., ``each r.e.\ set belongs to the 
class $\re$,'' although at times we will also use the 
terms RE and coRE 
themselves as adjectives.

$\epsilon$ 
will denote the empty string, and we use ``lexicographical''
in its standard computer science sense, 
e.g., $\epsilon \leqlex
0 \leqlex 1 \leqlex  00 \leqlex \cdots$.
$\nextstring(x)$ will denote the lexicographical successor of $x$,
e.g., $\nextstring(11)=000$.
For any set $A \subseteq \sigmastar$ and any string
$x \in \sigmastar$, $A^{\leq x}$ denotes all strings in $A$ that are
lexicographically less than or equal to $x$.  
$A^{<x}$ and $A^{\geq x}$ are defined analogously.
We will use these notations even for $\sigmastar$ itself,
e.g., if $x$ is the string $10$ then $(\sigmastar)^{\leq x}$ is 
the set $\{\epsilon,0,1,01,10\}$.
For each finite set $A$, $\|A\|$ will denote the cardinality of~$A$.
The function $\pair{ \cdot, \cdot}$
will denote a 
fixed, standard, recursive pairing function,
i.e.,
a recursive bijection between $\sigmastar \times \sigmastar$ and $\sigmastar$.

We say that $A_1$ and $A_2$ are \emph{(recursively) isomorphic}, denoted
$A_1 \equiv_{\textit{iso}} A_2$, exactly if there is a 1-to-1, onto,
total recursive function $f$ from $\sigmastar$ to $\sigmastar$ such
that $f(A_1) = A_2$.  For any reducibility $\leq_\alpha$, 
we say $A \equiv_\alpha B$ exactly if $A \leq_\alpha B$ and 
$B \leq_\alpha A$.  When $\leq_\alpha$ is reflexive 
and transitive, $\equiv_\alpha$ will be an equivalence relation,
and each equivalence class with respect to 
$\equiv_\alpha$ 
is said to be an \emph{$\leq_\alpha$ degree}.
The reducibilities whose degrees will 
be discussed during the rest of this paper
are
recursive many-one reductions ($\leq_m$), which we 
will also refer to simply as many-one reductions,
and recursive 1-truth-table
reductions ($\leq_{\textit{1-tt}}$), which we will also refer to 
simply as 
1-truth-table reductions. As is typical,
we will refer to
the
$\leq_m$ degrees 
and 
$\leq_{\textit{1-tt}}$ degrees as, respectively, 
many-one degrees and 1-truth-table degrees.
(Rather than define here the machinery of truth-table
reductions,
suffice it to say that the following 
is a true statement: $A$ 1-truth-table
reduces to $B$ exactly if $A$ Turing reduces to $B$ 
via a recursive transducer that on each input makes at most one query 
to $B$).
All recursive sets belong to a single 
1-truth-table degree, which in fact is exactly $\rec$.
All many-one degrees except the 
somewhat pathological many-one degree $\{\emptyset\}$ contain
infinite sets.

The \emph{Myhill Isomorphism Theorem}~\cite{myh:j:creative}
(or see~\cite[pp.~24]{soa:b:degrees})
states that $A\equiv_{\textrm{iso}}B \Longleftrightarrow 
A \equiv_{1} B$, where $\leq_1$ denotes (recursive) 1-to-1 reductions.
Though it is not a standard nickname, for 
convenience we will use the 
term 
\emph{Myhill's Corollary} to refer 
to the result 
that 
all RE-complete (with respect 
to many-one recursive reductions) sets are recursively 
isomorphic; and we will also refer to as Myhill's Corollary
the fact, semantically identical, that 
all coRE-complete (with respect 
to many-one recursive reductions) sets are recursively 
isomorphic.  
(Myhill's Corollary is well-known to with some argumentation
follow from
the Myhill Isomorphism Theorem, see, e.g., 
\cite[pp.~42--43]{soa:b:degrees} 
or 
\cite[Theorem III.6.6~+~Corollary~III.7.14]{odi:b:classical-recursion-theory}.)
$\K$ as mentioned earlier 
will denote the RE-complete set $\{x \condition x \in L(M_x)\}$,
where $L(M_i)$ denotes the language accepted by $M_i$, and 
$M_1,M_2,M_3,\ldots$ (or the same using strings as the subscripts under the 
standard correspondence between positive natural numbers and strings)
is a fixed, standard enumeration of (language-computing) Turing 
machines.

A set $A$ is a \emph{cylinder} exactly if for some $B$ it holds that 
$A \equiv_{\textit{iso}} B \times \sigmastar$.
We say that a set is $\core$ cylinder exactly if it is 
in coRE and is a cylinder.

We now define the class of compressible sets.  Our definition is the
precise analogue of the notion of P-compression of Goldsmith,
Hemachandra, and Kunen~\cite{gol-hem-kun:j:address}, except since we
will be studying the recursion-theoretic case we have removed the
requirement that the function be total and polynomial-time
computable.  Thus we are capturing the notion of a function
that when restricted to $A$ creates a total (on $A$), 
1-to-1, onto mapping to 
$\sigmastar$.
\begin{definition}[Compressible sets]\label{d:compressible}~
\begin{enumerate}
\item\label{d:compressible:p1} 
Given a set $A\subseteq \sigmastar$, we say that 
a (possibly partial) function $f$ is \emph{a compression 
function for $A$} exactly if
\begin{enumerate}
\item $\domain(f) \supseteq A$,
\item $f(A) = \sigmastar$, and 
\item 
$(\forall a \in A)(\forall b \in A)[ a\neq b \implies f(a) \neq f(b)]$.
\end{enumerate}
\item\label{d:compressible:p2}  
Let $\calf$ be any class of (possibly partial) functions
mapping from $\sigmastar$ to $\sigmastar$.
A set $A$ is \emph{$\calf$-compressible} exactly if 
$(\exists f \in \calf)[f \text{ is a compression function for } A]$.
\item\label{d:compressible:p3}  For each $\calf$ as above,
$\calf\text{-compressible} = \{ A \condition A \text{ is }
\calf\text{-compressible}\}$.
\item\label{d:compressible:p4}
 For each $\calf$ as above and each $\calc \subseteq 2^{\sigmastar}$,
we say that $\calc$ is \emph{$\calf$-compressible} exactly if 
$(\forall A \in \calc)[\text{If } A \text{ is an infinite set, then }
A \text{ is } \calf\text{-compressible}]$.
\end{enumerate}
\end{definition}

In a slight notational overloading, the above definition uses
$\calf\text{-compressible}$ both as an adjective and to represent the
corresponding class of sets.  Note that the above definition does not
constrain what $f$ does on elements of $\overline{A}$.  $f$ can be
undefined on some or all of those, and if it is defined on some of
those, note that that will make the overall map be non-1-to-1.  Of
course, $f$ may otherwise be constrained to be total, e.g., when we
speak of $\frec$-compressible sets, the $f$ involved must be total due
to the definition of $\frec$.  

No finite set has a compression function, since a finite set doesn't
have enough strings in it to map \emph{onto} $\sigmastar$.  This is
why part~\ref{d:compressible:p4} of the above definition
defines a class of sets as being compressible under a certain type of
function if all the class's infinite sets are thus compressible.
Though we could rig 
even parts~\ref{d:compressible:p2} 
and
\ref{d:compressible:p3} of
the definition of compression 
(rather than just 
part~\ref{d:compressible:p4})
to give 
finite sets a free pass, the given definition in 
each of these choices 
is exactly matching the long-standing, analogous 
complexity-theoretic definitions~\cite{gol-hem-kun:j:address}.
When we do wish to speak of the compressible sets augmented
by the finite sets, we will do so explicitly using the 
following:
$$\calf\text{-compressible}' =_{\text{def}} \calf\text{-compressible} \cup
\{ A \subseteq \sigmastar \condition A \text{ is finite}\}.$$

To get a sense of how compression works in a simple case, let 
us note the following.
\begin{proposition}\label{p:106.1}
$\re$
is $\fre$-compressible.
\end{proposition}
\begin{proof}
Let $A$ be any r.e.\ set.
Since $A$ is
r.e., there 
is an enumerating Turing machine that enumerates $A$
without repetitions.  Our $\fre$ compression function for $A$ will
map the $i$th enumerated string to the $i$th string in $\sigmastar$,
and will be undefined on all strings that are never enumerated (i.e.,
that belong to $\overline{A}$).  
\end{proof}

Compression of $A$ implies that in the image of the compression function
on $A$ we leave no holes: $f(A) = \sigmastar$.  That overall niceness
however does not imply that $f$ will never map any string in $A$ to a
lexicographically larger string.  $f$ certainly can, though the more
often it does so, the more often other strings in $A$ will need to map
to lexicographically smaller strings, to prevent any ``holes'' in the
image of $A$.  The more demanding notion called ranking, however,
\emph{does} ensure that no string in $A$ will map to a
lexicographically larger string.

Ranking is a particularly nice type of compression---compression that
simply maps the $i$th string in $A$ to the $i$th string in
$\sigmastar$.  There are three slightly differing versions of ranking,
depending on what one requires regarding what happens on inputs that
are not in $A$.  The following definition follows the one of those
that handles this analogously with the way it is handled in
compression, e.g., on inputs that are not in $A$ we allow 
the function to map to any strings it
wants, or even to be undefined.  
(Informally put, the compression function can ``lie'' or can be 
undefined on inputs $x\not\in A$.)
Hemachandra
and Rudich~\cite{hem-rud:j:ranking} (for the complexity-theoretic
case) defined this notion and called it ``weak ranking.'' However, to keep
our notations for compression and ranking in harmony with each other,
we will in this paper consistently refer to this simply as
``ranking.''\footnote{\label{f:other-ranking}The other two approaches to handling
  $\overline{A}$ have a
behavior is not too interesting in the recursion-theoretic world.  (In
the complexity-theory world, due to the work of Goldberg and
Sipser~\cite{gol-sip:j:compression} and Hemachandra and
Rudich~\cite{hem-rud:j:ranking}, it is known that for each of the
three notions, one has that all P sets are polynomial-time rankable
under that notion exactly if all $\#P$ functions---i.e., the counting
version of NP---are polynomial-time computable.)  The other two
notions are (a)~to additionally require that on members of
$\overline{A}$ the function either is undefined 
or states that they are not
members of $A$ (this notion's analogue is called ``ranking''
in~\cite{hem-rud:j:ranking}), or (b)~to additionally require that on
members of $\overline{A}$ the function computes and outputs
$\| A^{\leq x}\|$ (this notion's analogue is called ``ranking''
in~\cite{gol-sip:j:compression} and is called ``strong ranking''
in~\cite{hem-rud:j:ranking}).  However, under each of these notions,
with respect to either of $\frec$ or $\fre$, the class of sets thus
rankable is exactly the recursive sets; we include a proof of this 
in 
Appendix~\ref{sapp:defs}.  
Thus these two
notions, though interesting in the complexity-theoretic study of
ranking, are not interesting in the recursion-theoretic study of
ranking.}

\begin{definition}[Rankable sets]\label{d:rankable}~
\begin{enumerate}
\item Given a set $A\subseteq \sigmastar$, we say that 
a (possibly partial) function $f$ is \emph{a ranking 
function for $A$} exactly if
\begin{enumerate}
\item $\domain(f) \supseteq A$, and 
\item if $x \in A$, then $f(x) = 
\| A^{\leq x}\|$.  (That is, if $x$ is the $i$th string in $A$, 
then $f(x)$ is the $i$th string in $\sigmastar$.)
\end{enumerate}
\item 
Let $\calf$ be any class of (possibly partial) functions
mapping from $\sigmastar$ to $\sigmastar$.
A set $A$ is \emph{$\calf$-rankable} exactly if 
$(\exists f \in \calf)[f \text{ is a ranking function for } A]$.
\item For each $\calf$ as above,
$\calf\text{-rankable} = \{ A \condition A \text{ is }
\calf\text{-rankable}\}$.
\item\label{d:rankable:part-class}
 For each $\calf$ as above and each $\calc \subseteq 2^{\sigmastar}$,
we say that $\calc$ is \emph{$\calf$-rankable} exactly if 
$(\forall A \in \calc)[
A \text{ is } \calf\text{-rankable}]$.
\end{enumerate}
\end{definition}

For example, clearly every recursive set is 
$\frec$-rankable by brute force.  However, we will later see
that, in contrast, some infinite r.e.\ sets are not even
$\frec$-compressible.

Aside from the quirk that finite sets cannot be compressible,
rankability clearly implies compressibility.  And of course,
every total recursive function is a partial recursive function.
So we have the following trivial containments.
\begin{proposition}\label{p:housekeeping-1}
\begin{enumerate}
\item \label{p:housekeeping-1-a}
$(\forall \calf) [ \calf\text{-rankable} \subseteq \calf\text{-compressible}']$.
\item 
$\frec\text{-rankable} \subseteq \fre\text{-rankable}$.
\item 
$\frec\text{-compressible} \subseteq \fre\text{-compressible}$.
\end{enumerate}
\end{proposition}

\section{Ranking}\label{s:ranking}
In Footnote~\ref{f:other-ranking} we noted that, for the ranking
variants where the ranking function's behavior on the complement is
constrained, the class of things that can be $\frec$ ranked, or even $\fre$
ranked, (in that variant) is precisely $\rec$, the recursive sets.

In contrast with those variants, we now show that arbitrarily complex
sets are $\frec$-rankable.  So, certainly, some $\frec$-rankable sets
are not recursive.

\begin{theorem}\label{t:103.1}
Every 1-truth-table degree 
contains an $\frec$-rankable set.
\end{theorem}

\begin{proof}%
  Let $\oplushat$ be defined by
  $A \oplushat B = \{x0\condition x\in A\} \cup \{x1\condition x\in
  B\}$, i.e. this is the standard ``join'' (aka ``disjoint union,''
  aka ``marked union''), except the marking bit is the low-order bit
  rather than as is standard the high-order bit.  Note that for
  \emph{any} set $A\subseteq \sigmastar$, $A\oplushat \overline{A}$ is
  $\frec$-rankable (indeed, it is even Logspace-rankable) by the
  function defined by
$$f(\epsilon)=\epsilon,~ f(z0)=z, \text{ and } f(z1)=z,$$
  since for each $x$ exactly one of $x0$ and $x1$ is in
  $A\oplushat \overline{A}$, and $\epsilon \not\in A\oplushat
  \overline{A}$.
\end{proof}

\begin{corollary}
There exist sets $A$ that are not in the arithmetical hierarchy 
yet are $\frec$-rankable.
\end{corollary}

However, it follows from Theorem~\ref{t:109.2} of the next section---which
establishes that $\rec = \re \cap \frec\text{-compressible}'$---that 
Theorem~\ref{t:103.1} cannot be improved from 1-truth-table
degrees to many-one degrees.

Theorem~\ref{t:103.1} shows that $\frec$-rankable sets occur
everywhere.
Nonetheless,  we show as Theorem~\ref{t:123.1}
that the \emph{non}-$\frec$-rankable sets also occur
everywhere.  
Theorem~\ref{t:25.2BK54} 
notes that for the 
case of r.e.\ sets, $\fre$-rankability even implies 
decidability,
thus all sets in $\re - \rec$ are 
non-$\frec$-rankable.

\begin{theorem}\label{t:123.1}
Every 1-truth-table degree except that of the recursive sets 
contains a set that is $\frec$\tcomp\ but not $\frec$-rankable.
\end{theorem}

\begin{proof}%
Let $A$ be an arbitrary nonrecursive set.
Let $s_0, s_1, s_2, ...$ enumerate $\sigmastar$ in lexicographical order.
Define 
$$B = \{ s_{4i} \condition i\geq 0 \} \cup \{ s_{4i+1} \condition s_i \in A \} \cup \{ s_{4i+2}\condition i \geq 0 \} \cup \{ s_{4i+3} \condition s_i \in \overline{A} \}.$$
So $B$ consists of a pattern that repeats every four strings. 
Namely, the first and third strings are always in, and exactly one of the second and fourth is in.
Then $A \equiv_{\textit{1-tt}} B$ and $B$ is $\frec$\tcomp\ by the map $f$ defined by
$$f(s_{4i}) = s_{3i},$$
$$f(s_{4i+1}) = s_{3i+1},$$
$$f(s_{4i+2}) = s_{3i+2},\text{ and}$$
$$f(s_{4i+3}) = s_{3i+1}.$$
However, if $B$ were $\frec$\trank, then $A$ would be recursive.
Why?
If $g$ 
is an $\frec$ ranking function for $B$,
then it holds that
$$s_i \in A \iff g(s_{4i+2}) - g(s_{4i}) = 2.$$
Since $A$ is not recursive, $B$ cannot be $\frec$\trank.
\end{proof}

\begin{theorem}\label{t:25.2BK54}
Every r.e.\ $\fre$\trank\ set is recursive.  (Equivalently,
$\rec = \re \cap \fre\trank
= \re \cap \frec\trank$.)
\end{theorem}

\begin{proof}%
This proof is similar in flavor to the proof that comprises 
the whole of 
Appendix~\ref{sapp:defs}, except in that proof but not here 
one has a model in which the ranker is not allowed to output 
``lies'' as to the rank of nonmembers of the set, and here but not 
there we have the assumption that the set is r.e.

Let $A$ be an r.e.\ $\fre$\trank\ set.  
Let $f$ be an $\fre$-ranking function for $A$.
Since $A$ is r.e., there exists an 
enumerating Turing machine, $E$, for $A$, and without loss of 
generality, we assume that $E$ enumerates the elements of $A$ without
repetition.  
If $A$ is finite, then $A$ is recursive, so we in the 
following consider just the case that $A$ is infinite.
Here is our algorithm, which will always halt, to 
decide membership in $A$.  On arbitrary input $x$, for which we wish 
to test whether $x \in A$, start running the enumerating machine $E$.
Each time the machine outputs an element, run $f$ on that element to 
determine the correct rank of that element (since the elements output by $E$
all belong to $A$, $f$ halts on each and outputs the correct rank value, e.g.,
if the string is the seventh string in $A$, the function $f$ will output the 
lexicographically seventh string in $\sigmastar$).  Each time we thus obtain 
a rank, check to see if either: (a)~the string just output by $E$ is $x$, 
in which case accept as $x\in A$, (b)~the ranker has mapped some string $y$
output by $E$ and satisfying $y \glex x$ to the string $\epsilon$ (i.e.,
has declared it to be the lexicographically least string in $A$), in which
case reject as $x \not\in A$, or (c)~the ranker has mapped some two strings,
$y$ and $y'$---such that $y \llex x \llex  y'$
and both $y$ and $y'$ have by now have been output by $E$---to 
outputs $f(y)$ and $f(y')$
such that $f(y')$ is the lexicographical successor in $\sigmastar$ of $f(y)$,
in which case reject as $x \not\in A$.
\end{proof}

\begin{corollary}\label{c:25.3BK54}
No $\re$-complete set is $\fre$\trank.
\end{corollary}

Let us now turn to seeing how the co-r.e.\ sets---especially the
coRE-complete sets---interact with 
$\frec$-rankability and
$\fre$-rankability.  

First, though, let us notice that 
for the co-r.e.\ sets, 
$\frec$-rankability and
$\fre$-rankability precisely coincide
(though unlike the case---see Theorem~\ref{t:25.2BK54}---of the RE sets, that
as shown by Corollary~\ref{c:167.3-dan} 
is not due to them both collapsing to the recursive sets).

\begin{theorem}\label{t:169.1}
$\core \cap \frec\trank = 
\core \cap \fre\trank$.
\end{theorem}

\begin{proof}%
  Let $A$ be a set in $\core \cap \fre\trank$.  We give an
  $\frec$-ranker for $A$.  Namely, on input $x$, run both the
  $\fre$-ranker $f$ and an enumerator for $\overline{A}$, dovetailed,
  until we either get a value for $f(x)$ from the ranker or we see the
  enumerator state that $x \not\in A$.  If the former, output that value,
  and if the latter, output any fixed string, e.g., 101010.
\end{proof}

Next we give the following theorem, which implies
us our desired result about coRE-complete sets, and more.

\begin{theorem}\label{t:167.1-dan}
If $A$ is an $\fre$\trank\ co-r.e.\ set that has an 
infinite r.e.\ 
subset, then $A\in\rec$.
\end{theorem}

\begin{proof}%
Let $A$ be as in the theorem's hypothesis.
Let $s_0, s_1, s_2...$~enumerate $\sigmastar$ in lexicographical order. 
Let $E$ be an enumerating 
Turing machine 
without repetitions for $\overline{A}$ and 
let $F$ be an enumerating Turing machine 
for
an infinite r.e.\ subset of $A$. 
Suppose $g$ is a ranking function for $A$. 
In light of Theorem 10, $A$ is $\fre$\trank\ $\iff A$ is $\frec$\trank, so  
w.l.o.g.~we 
assume that $g \in \frec$.

Then the following procedure decides whether $x \in A$. 
Run $F$ until it enumerates some string $s_n \glex x$. 
Compute $g(s_n)$. 
Since $s_n \in A$, 
$(\sigmastar)^{\leq s_n}$ 
is composed of $g(s_n)$ members of $A$ and $n-g(s_n)$ members of 
$\overline{A}$.
Run $E$ until it enumerates $n-g(s_n)$ strings in $(\sigmastar)^{\leq s_n}$.
If %
$x$ is one of those 
$n-g(s_n)$ strings, then 
we know that $x \not \in A$, and 
otherwise we know that $x \in A$.
\end{proof}

\begin{corollary}\label{c:167.3-dan}
No $\core$-complete set is $\fre$-rankable.
\end{corollary}

\begin{proof}%
  This follows directly from Theorem~\ref{t:167.1-dan}, in light of
  Post's~\cite{pos:j:re} early result that every coRE-complete set has
  an infinite r.e.\ subset.\footnote{Post's result is trivial to see
    these days, using the fact (what we are calling Myhill's Corollary) 
    that all coRE-complete sets are
    isomorphic, but Post didn't have the benefit of Myhill's Corollary
    and thus proved his result directly.  We note in passing the
    following slight extension of Post's result, since we could not
    find it in the existing literature: Every coRE-complete set has an
    infinite $\re - \rec$ subset.  To see this, just note that the 
    coRE-complete set $\overline{\K}\oplus\sigmastar$
    has the infinite $\re-\rec$ subset $\{1y\condition y\in\K\}$,
    and so by Myhill's Corollary the claim follows.}
\end{proof}

\begin{corollary}\label{c:cylinder}
Every 
$\fre$\trank\, co-r.e.\ cylinder is recursive.
\end{corollary}

\begin{proof}%
Each finite co-r.e.\ set is recursive.
Each infinite co-r.e.\ cylinder has an infinite 
recursive subset (if the cylinder is 
recursively isomorphic to $B \times \sigmastar$
via recursive isomorphism function $h$,
then for any fixed $x\in B$, we have that
the set $h^{-1}(\{ (x,y) \condition y\in\sigmastar\})$ 
is such an infinite recursive set), and so we are done 
by Theorem~\ref{t:167.1-dan}.
\end{proof}

Though by Theorem~\ref{t:103.1} there are $\frec$-rankable sets in the
1-truth-table degree of $\K$,
we also know
that none of those sets can be RE-complete or coRE-complete.
The impossibility of them being coRE-complete follows from 
Corollary~\ref{c:167.3-dan}, which indeed precludes 
even 
$\fre$-rankability.
The impossibility of them being RE-complete will follow from 
the coming Corollary~\ref{c:111.15}, which indeed precludes even 
$\frec$-compressibility.  We state this as the following corollary.

\begin{corollary}
Although the 1-truth-table degree of the $\re$-complete sets contains 
$\frec$\trank\ sets, no $\re$-complete or $\core$-complete sets 
are 
$\frec$\trank.
\end{corollary}

\section{Compression}\label{s:compression}
Proposition~\ref{p:106.1} shows that every infinite
r.e.\ set is $\fre$-compressible.  We note in passing 
that from that and 
Theorem~\ref{t:25.2BK54}
we immediately have the following.
\begin{proposition}\label{p:122.1}
There exist r.e.\ sets---in fact, all of $\re - \rec$---that 
are $\fre\tcomp$ yet are not 
$\frec\trank$ or even $\fre\trank$.
\end{proposition}
We will soon see that 
in that proposition
$\fre$\tcomp\ cannot 
be improved to $\frec$\tcomp.

The following result shows that for $\frec$ compression (and even for
Logspace compression, if one looks inside the proof of
Theorem~\ref{t:103.1}), compressible sets 
exist in every 1-truth-table degree.
(This result is a corollary to the proof of 
Theorem~\ref{t:103.1}---it follows, in light of 
Proposition~\ref{p:housekeeping-1}'s part~\ref{p:housekeeping-1-a}, 
from the fact that the sets constructed
in the proof of 
Theorem~\ref{t:103.1}
are infinite.)

\begin{corollary}\label{c:107.5}
Every 1-truth-table degree 
contains an $\frec$\tcomp\ set.
\end{corollary}

Can we improve Proposition~\ref{p:122.1}'s claim from $\fre$\tcomp\ to $\frec$\tcomp?
Can we improve Corollary~\ref{c:107.5}'s claim from 1-truth-table
degrees to many-one degrees (to avoid this being trivially 
impossible due to the pathological many-one degree that 
contains only the empty set, what we actually are asking is 
whether we can change Corollary~\ref{c:107.5} to 
``every many-one degree other than 
$\{\emptyset\}$
contains 
an $\frec$\tcomp\ set'')
or, and this would not be an improvement
but rather would be an incomparable claim, can we change the claim as
just mentioned to 
all many-one degrees other than 
$\{\emptyset\}$
if we in addition 
restrict our attention just to the r.e.\ degrees?
Or can we perhaps hope to show that 
$\re \subseteq \frec\tcomp'$?
The following result, observed without proof in the 
conclusions section of~\cite{gol-hem-kun:j:address},
implies 
that the answer to each of these questions
is ``no'';
$\frec$ compression 
is impossible for sets in $\re - \rec$.

\begin{theorem}[\cite{gol-hem-kun:j:address}]
\label{t:109.2}
$\rec = \re \cap \frec\tcomp'$.
\end{theorem}

\begin{proof}%
The $\subseteq$ direction is immediate.
Let us show the $\supseteq$ direction.  Let $A\in \re\cap \frec\tcomp'$.
If $A$ is finite then certainly $A\in\rec$, so only the case of infinite 
$A$ remains.  
$A\in\core$, since, where $f$ 
is the $\frec$-compressor function for our infinite r.e.\ set $A$, 
$$\overline{A} = \{x\condition (\exists y) [ y\in A \land y\neq x \land 
f(y)=f(x)]\}.$$
So $A$ is r.e.\ and co-r.e., and thus is recursive.
\end{proof}

\begin{corollary}\label{c:111.15}
  No set in $\re - \rec$ is $\frec$\tcomp.
  In particular, no $\re$-complete set is $\frec$\tcomp.
\end{corollary}

Corollary~\ref{c:111.15} follows immediately from 
Theorem~\ref{t:109.2} and so needs no proof.
Corollary~\ref{c:111.15} brings out a clear asymmetry, regarding
compression, between
$\re$ and $\core$: no $\re$-complete set is 
$\frec$\tcomp, but as 
we will soon see Corollary~\ref{c:4.1/4.2-dan},
all coRE-complete sets are
$\frec$\tcomp.

Does Theorem~\ref{t:109.2} remain true if we change $\frec$ to $\fre$?
We already know that answer is ``no''; in fact, from
Proposition~\ref{p:106.1}
not only do we have that 
$\re \cap \fre\tcomp \not\subseteq \rec$,
we indeed have that 
$\re \cap \fre\tcomp' = \re$.

As promised above, although the RE-complete sets---indeed, all sets in
$\re - \rec$---are not $\frec$\tcomp, we will now establish that all
coRE-complete sets \emph{are} $\frec$\tcomp.  Although we can prove that
directly,\footnote{For the reader who might like a direct, simple proof
of 
Corollary~\ref{c:4.1/4.2-dan}, we include here such a construction.
Consider the set 
$A = \{\pair{x,\epsilon} \condition x \in \overline{\K}\}
\cup
\{ \pair{x,\nextstring(y)} \condition M_x(x)
\text{ accepts in exactly $y$ steps}\}$,
where $\nextstring(y)$ denotes the string immediately after $y$ 
in lexicographical order.  $A$ clearly is coRE-complete
and, since for each $x$ there is exactly one $y$ such that 
$\pair{x,y}\in A$, $A$ is $\frec$-compressible 
via the function $f(\pair{x,y}) = x$.  In fact, 
the proof of Theorem~\ref{t:dan-160522-Beta1} is simply a more 
flexible version of this idea.}
we will instead state a more general result (Theorem~\ref{t:dan-160522-Beta1})
of interest in
its own right, and which yields the claim as a corollary.

\begin{proposition}\label{p:isomorphisms}
$\frec$\tcomp\ and 
$\fre$\tcomp\ are each closed under recursive 
isomorphisms.
\end{proposition}

Proposition~\ref{p:isomorphisms} is immediate
and needs no proof.

\begin{theorem}\label{t:dan-160522-Beta1}
Every $\core$ cylinder except $\emptyset$ is 
$\frec$\tcomp.
\end{theorem}

\begin{proof}%
Let $A$ be co-r.e.\ and a nonempty cylinder.
Let $s_0, s_1, s_2, ...$ enumerate $\sigmastar$ in lexicographical order.
Let $E$ be an enumerating Turing machine 
without repetitions
for  $\overline{A}$.
Define $L_A = \{\pair{x,\epsilon} \condition x \in A\} \cup \{\pair{x,s_i} \condition i \geq 1 \,\,\land\,\, x$ is the $i^{\rm{}th}$ string enumerated by $E\}$.
Then $L_A$ is $\frec$\tcomp\ by projection onto the first coordinate.

We claim that $A$ is recursively isomorphic to $L_A$. Why?

Clearly $A \leq_1 L_A$ by mapping $x$ to $\pair{x,\epsilon}$.
Now let us show that 
$L_A \leq_1 A$.
Fix strings $x_0 \in A$ and $x_1 \not \in A$. Then the following gives a many-one reduction from $L_A$ to $A$:
\begin{quotation}
\noindent On input $\pair{x,s_i}$, if $s_i = \epsilon$, output $x$. Otherwise, check if $x$ is the $i^{\rm{}th}$ string enumerated by $E$. If so, output $x_0$. Otherwise, output $x_1$.
\end{quotation}
So $L_A \leq_m A$. Since $A$ is a cylinder, 
it follows 
(by~\cite[p.~89]{rog:b:rft}) that 
$L_A \leq_1 A$. 
Since 
$A \leq_1 L_A$ and 
$L_A \leq_1 A$, by 
the Myhill Isomorphism Theorem $L_A$ is recursively isomorphic to $A$, so 
by Proposition~\ref{p:isomorphisms} $A$ is $\frec$\tcomp.
\end{proof}

Note that the set 
$\overline{\K}_{\text{cyl}} =_{\text{def}}
\{\pair{a,b} \condition a\in \overline{\K} \land
b\in\sigmastar\}$ 
is clearly a
coRE-complete cylinder (since it is trivially 
recursively isomorphic to the two-dimensional
set $\overline{\K}\times \sigmastar$ via 
$\pair{a,b} \mapsto (a,b)$).
So by 
Theorem~\ref{t:dan-160522-Beta1}
we have that that coRE-complete set is 
$\frec$\tcomp.  But since 
$\frec$\tcomp\ is, as Proposition~\ref{p:isomorphisms} notes, 
clearly closed under recursive 
isomorphisms (as already has been analogously noted 
before for the case of polynomial-time 
compressibility and polynomial-time 
isomorphisms~\cite{gol-hem-kun:j:address}), 
and since by Myhill's Corollary all coRE-complete sets are 
recursively isomorphic to 
$\overline{\K}_{\text{cyl}}$,
we have that 
all coRE-complete sets are 
$\frec$\tcomp.  We summarize this as the following corollary, 
whose claim appeared without proof in~\cite{gol-hem-kun:j:address}.

\begin{corollary}[Stated without proof in~\cite{gol-hem-kun:j:address}]
\label{c:4.1/4.2-dan}
All $\core$-complete sets are 
$\frec$\tcomp.  
\end{corollary}

Fix a standard, nice indexing 
(naming scheme)---$\phi_1,\phi_2,\phi_3,\ldots$---for
the partial recursive functions.
A set $A$ is an \emph{index set} exactly if 
there exists a (possibly empty) 
collection $\calf'$ of partial recursive functions 
such that 
$A = \{i \condition \phi_i \in \calf'\}$.
(Since all our sets are over $\sigmastar$, we are implicitly
associating the $i$th positive natural number with the $i$th 
string in $\sigmastar$, so that our index sets are type-correct.)

Since all index sets 
are cylinders (see~\cite[p.~23]{soa:b:degrees}; \cite{soa:b:degrees}'s 
definition of cylinders is well-known---see
\cite[Theorem~VIII(c)]{rog:b:rft}---to be equivalent to the definition
given in our Section~\ref{s:definitions}),
and thus all co-r.e.\ index sets are co-r.e.\ cylinders,
Theorem~\ref{t:dan-160522-Beta1} implies that 
all co-r.e.\ index sets except (the finite, and thus
not compressible index set) $\emptyset$ are $\frec$\tcomp.
\begin{corollary}\label{c:dan-160522-Beta2}
All coRE index sets except $\emptyset$ are 
$\frec$\tcomp.  
\end{corollary}

On the other hand, 
by diagonalization 
we can 
build a set, even one recursive in $\K$,
that is not compressible even by any $\fre$ function.

\begin{theorem}\label{t:151.2(b)-dan}
$\deltatwozero \not\subseteq \fre\tcomp'$.
\end{theorem}

\begin{proof}%
Fix a standard enumeration of Turing machines, $M_1, M_2, M_3, ...$,
with each machine viewed as 
computing a partial recursive function.
$\phi_i$ will denote the partial recursive function computed by $M_i$.
We will explicitly construct a set $A$ that belongs 
to $\Delta_2^0$ but that is not $\fre$\tcomp.
This will be done by a stage construction.
At stage $i$, we will define a set $A_i$ and a string $w_i$.
$A$ will be defined as $\bigcup_{i \geq 0} A_i$.
We will ensure that $A_i \subseteq A_{i+1}$ and $A^{< w_i} = A_i^{< w_i}$.
That is, after stage $i$, all strings lexicographically preceding $w_i$ 
will be fixed---their membership/nonmembership in $A$ will not 
be changed by later stages.
At each stage $i$, $i\geq 1$,
at least one string will be added to $A$, in order to 
ensure that $A$ ultimately becomes infinite,
and $\phi_i$ will be eliminated as a 
$\fre$-compressor for $A$.

We start by setting $A_0$ to 
be 
$\emptyset$ and 
$w_0$ to be $\epsilon$.  We then do stage 1, then stage 2, etc.

At stage $i$, $i \geq 1$, first check whether $\phi_i$ 
is injective when restricted to $A_{i-1} \cup (\sigmastar)^{\geq w_{i-1}}$. 
This is an r.e.~condition, since we are looking for a pair $(x,y)$ 
with $x,y \in \domain(\phi_i)$, $x  \neq y$, $x,y \in A_{i-1} 
\cup (\sigmastar)^{\geq w_{i-1}}$, and $\phi_i(x) = \phi_i(y)$.
If such a pair exists (which our $\deltatwozero$ process can easily test),
then 
we set $A_i$ to be $A_{i-1} \cup \{x,y,w_{i-1}\}$, we set $w_i$
to be $\nextstring(\max(x,y,w_{i-1}))$, and we go to stage $i+1$.
(We added $w_{i-1}$ to ensure that we always add a string---even 
in the case that $x,y \in A_{i-1}$.)
Since 
$x,y \in A_i$ (and thus $x,y \in A$)
and $x \neq y$, we 
have ensured that for two strings in $A$, namely $x$ and $y$, 
$\phi_i$ maps to the same output.  So $\phi_i$ has been eliminated 
as a 
potential 
compressor for $A$.

However, if we cannot find such a pair $(x,y)$, then we will 
freeze a string out of $A$ in such a way as to permanently ensure that 
$\phi_i$ is not surjective.
In particular, 
check whether $\phi_i((\sigmastar)^{\geq w_{i-1}}) \neq \emptyset$. 
This is again an r.e.~test.  If the test determines that 
$\phi_i((\sigmastar)^{\geq w_{i-1}}) \neq \emptyset$, then 
there is an $x \in (\sigmastar)^{\geq w_{i-1}}$ such that $\phi_i(x)$ 
is defined.
Let $x$ be the lexicographically smallest string 
in 
$(\sigmastar)^{\geq w_{i-1}}$
such that this holds. Such an $x$ can 
be found by further r.e.~queries (within our $\deltatwozero$ process).
Then set $A_i$ to be $A_{i-1} \cup \{\nextstring(x)\}$, 
set $w_i$ to be $\nextstring(\nextstring(x))$, and go to stage $i+1$.
Note that we have ensured that 
$\phi_i(x) \not \in \phi(A)$ (the reason we know that no string in 
$A_{i-1}\cup (\sigmastar)^{\geq w_{i-1}}$ can map to $\phi_i(x)$
is that if so we \emph{would} have had a pair $(x,y)$ of the 
sort sought above), 
and so $\phi_i(A) \neq \sigmastar$, and 
so 
we have ensured that $\phi_i$ is not a compressor for $A$.

In the last remaining case, we must 
have $\phi_i((\sigmastar)^{\geq w_{i-1}}) = \emptyset$.
So $\phi_i$ is only defined for finitely many strings, and thus 
cannot be a compressor function for $A$. So set $A_i$ to 
$A_{i-1} \cup \{w_{i-1}\}$, set $w_i$ to $\nextstring(w_{i-1})$,
and go to stage $i+1$.

Note that in all three cases, $A_i$ has at least one more element than
$A_{i-1}$, so $A$ will be infinite. And at stage $i$, we ensured that 
$\phi_i$ will not be a compressor for $A$. So $A$ is not $\fre$\tcomp,
since no partial recursive function is a compressor for it.
\end{proof}

\section{Conclusions and Open Problems}\label{s:open}
This paper defined and studied the recursion-theoretic analogues of
the complexity-theoretic notions of ranking and compression.  We
particularly sought to determine where rankable and compressible sets
could be found.  For example, we found that all coRE-complete sets are
recursively compressible but no RE-complete set is recursively
compressible, and
that no RE-complete or coRE-complete set is recursively (or even
partial-recursively) rankable.  Nonetheless, we showed that every
1-truth-table degree---even the one containing the RE-complete and the
coRE-complete sets---contains recursively rankable sets.  And we also
showed that every nonempty coRE cylinder is recursively compressible.

We mention some open issues that we commend to the interested reader.
We conjecture that there exist
infinite, co-r.e.\ sets that are not 
partial-recursively
compressible, although 
this paper 
establishes that only for the 
larger class $\deltatwozero$.  
Also, can one construct a set that is
$\fre$\trank\ but not $\frec$\trank, and if so, what is the smallest
class in which such a set can be constructed?  Note that by
Theorems~\ref{t:123.1} and~\ref{t:169.1}, separating 
$\fre$\trank\ from $\frec$\trank\ on any set in $\re \cup \core$ 
is impossible.
Some additional research 
directions and ideas for work 
building on the notions of 
the present paper 
can be found in~\cite{hem-rub:t3:rft-compression2}.


%

%

%
\appendix

\section{%
Deferred Proof from Section~\ref{s:definitions}}\label{sapp:defs}
In Footnote~\ref{f:other-ranking}, in
Section~\ref{s:definitions},
we claimed that under both the
``a'' and ``b'' variants of ranking mentioned in that footnote, and
for each of those under both $\frec$ and $\fre$ ranking functions, the
class of sets thus ranked is exactly the recursive sets.  We now prove
that.  It is immediately obvious that each recursive set is $\frec$
(the more restrictive of the two function classes) rankable even under
the ``b'' variant, which is the more restrictive of the two variants.
So all that remains is to show that each set that is 
$\fre$ rankable
under the ``a'' variant is recursive.  Let $A$ be a set that is 
$\fre$ rankable
under the ``a'' variant.  
If $A$ is finite, then trivially $A\in\rec$.
So let us consider the case where $A$ is infinite.
Let $f$ be an $\fre$ ranking function for $A$
of the variant 
``a'' sort.  
Let us quickly make clear what variant ``a'' means,
especially in the context of $\fre$ functions.
If on an input $f$ halts in an accepting 
state we view 
the string that is at that moment on its output tape 
(namely, from the left end of the output tape
up to but not including the leftmost blank cell) as the output of 
$f$, and if $f$ halts in a rejecting state we view it as
stating that the input is not in the set.  
On inputs $x \in A$, $f$ must output the string whose rank order
within $\sigmastar$ is the same as the rank order of $x$ within $A$.
On inputs $x \not\in A$, $f$ can either halt in a rejecting state or 
run forever (but it cannot halt in an accepting state, i.e., it cannot
output some string; this contrasts with
Definition~\ref{d:rankable}, which allows $f$ to even ``lie'' 
on inputs $x\not\in A$).
Here is the description of a procedure, which halts on every
input,
for testing whether $x\in A$.  In a standard dovetailing manner
(i.e., interleaved, e.g., running on the first string in $\sigmastar$ 
for one step,
then running on the first two strings in $\sigmastar$ 
for two steps each, then running
on the first three strings in $\sigmastar$ for three steps each, and so on),
run $f$ on every string in $\sigmastar$.  If $f(x)$ is ever 
computed in that process, we reject $x$ if $f(x)$ declares that $x$ 
is not in $A$ (recall that as noted above in variant ``a'' the ranker 
can declare the string to not be in $A$, in particular by 
halting in a rejecting 
state), and we accept $x$ otherwise.  Also, as the process goes on, if any 
string $y$ such that $y \glex x$ evaluates to the lexicographically 
first string in $\sigmastar$, namely $\epsilon$, then we reject $x$.
Also, as the process goes on, if for some strings $w$
and $y$ 
with $w \llex x \llex y$ and such that $f(y)$ and $f(w)$ have 
both evaluated, it holds that $f(y)$ evaluates to 
the lexicographical
successor of $f(w)$, then reject $x$.  

At least one of these cases must eventually occur.  Why?
If $x \in A$, eventually, $f$ will compute $f(x)$ and we will correctly
accept.
If $x \not\in A$, then there are two cases.  If $x$ is lexicographically
strictly less than the lexicographically first 
string $z$ in $A$, then eventually we will evaluate $f(z)$ to be 
$\epsilon$ and will correctly reject $x$.  Otherwise, eventually the 
strings in $A$ that are most closely lexicographically greater than 
(recall that we are here handling the case that $A$ is infinite,
so such a string must exist)
and less than $x$ 
will evaluate under $f$, at which point we will correctly reject $x$.

\end{document}